\documentclass[conference,letterpaper]{IEEEtran}

\usepackage[utf8]{inputenc} 
\usepackage[T1]{fontenc}
\usepackage{url}
\usepackage{ifthen}
\usepackage{cite}
\usepackage[cmex10]{amsmath} 
\usepackage[affil-it]{authblk}
\usepackage[usenames,dvipsnames]{xcolor}
\usepackage{amsfonts}
\usepackage{amsmath,amsthm,amssymb,dsfont}
\usepackage{enumerate}
\usepackage[english]{babel}
\usepackage{graphicx}	
\usepackage[caption=false]{subfig}
\usepackage{url}
\usepackage{todonotes}
\usepackage{bbm}
\usepackage{booktabs}
\usepackage{stmaryrd}
\usepackage{tikz}
\usetikzlibrary{chains}
\usetikzlibrary{fit}
\usepackage{pgflibraryarrows}		
\usepackage{pgflibrarysnakes}		
\usepackage{caption}

\usepackage{makecell}

\usepackage{diagbox}

\usepackage{epsfig}
\usetikzlibrary{shapes.symbols,patterns} 
\usepackage{pgfplots}
\pgfplotsset{compat=newest}

\usepackage[linesnumbered,ruled,vlined]{algorithm2e}
\usepackage{algorithmic}
\usepackage{hyperref}
\hypersetup{colorlinks=true,citecolor=blue,linkcolor=blue,filecolor=blue,urlcolor=blue,breaklinks=true}

\usepackage{nicefrac}
\usepackage{mathtools}
\usepackage{mathrsfs}
\usepackage{multirow}

\theoremstyle{plain}
\newtheorem{theorem}{Theorem}
\newtheorem{lemma}[theorem]{Lemma}

\newtheorem{corollary}[theorem]{Corollary}

\theoremstyle{definition}

\newtheorem{example}[theorem]{Example}

\newcommand*{\cA}{\mathcal{A}}
\newcommand*{\cB}{\mathcal{B}}
\newcommand*{\cC}{\mathcal{C}}

\newcommand*{\cS}{\mathcal{S}}
\newcommand*{\cT}{\mathcal{T}}

\newcommand*{\NN}{\mathbb{N}}

\newcommand*{\bigO}{\mathcal{O}}

\newcommand{\abs}[1]{\left|#1\right|}

\newcommand{\suppress}[1]{}

\SetKwInput{KwInput}{Input}                
\SetKwInput{KwOutput}{Output}              

\begin{document}

\title{On de Bruijn Covering Sequences and Arrays} 

\author{Yeow Meng Chee\IEEEauthorrefmark{1},
        Tuvi Etzion\IEEEauthorrefmark{1}\IEEEauthorrefmark{2},
        Hoang Ta\IEEEauthorrefmark{1},
        and Van Khu Vu\IEEEauthorrefmark{1}}

\affil{
\IEEEauthorrefmark{1}\small{Department of Industrial Systems Engineering and Management, National University of Singapore, Singapore}\\
\IEEEauthorrefmark{2}\small{Department of Computer Science, Technion --- Israel Institute of Technology, Haifa, 3200003 Israel}
}

\maketitle

\begin{abstract}
An $(m,n,R)$-de Bruijn covering array (dBCA) is a doubly periodic $M \times N$ array
over an alphabet of size $q$ such that the set
of all its $m \times n$ windows form a covering code with radius~$R$.
An upper bound of the smallest array area of an $(m,n,R)$-dBCA is provided using a probabilistic technique
which is similar to the one that was used for an upper bound on the length of a de Bruijn covering sequence.
A folding technique to construct a dBCA from
a de Bruijn covering sequence or de Bruijn covering sequences code is presented.
Several new constructions that yield shorter de Bruijn covering sequences
and $(m,n,R)$-dBCAs with smaller areas are also provided.
These constructions are mainly based on sequences derived from cyclic codes,
self-dual sequences, primitive polynomials,
an interleaving technique, folding, and mutual shifts of sequences with the same covering radius.
Finally, constructions of de Bruijn covering sequences codes are also discussed.
\end{abstract}

\section{Introduction}
\label{sec:introduction}
A de Bruijn sequence of order $n$ over an alphabet ${\Sigma=\{0,1,\ldots,q-1\}}$ is a cyclic sequence
of length $q^n$ such that each $n$-tuple appears exactly once as a substring
in the sequence. Such a sequence exists for any
alphabet~\cite{de1946combinatorial, IJgood46, dBEh51a}.
The de Bruijn sequences are important from both theoretical and practical points of view, e.g., see~\cite{lempel1985design,golomb2017shift,bruckstein2012simple,samatham1989bruijn,chee2022run,pages2005optimised}.
The proof for their existence is based on the existence of Eulerian circuits and Hamiltonian cycles
in the de Bruijn graph $G_{q,n}$ defined by~\cite{de1946combinatorial, IJgood46}.
The graph is directed and has $q^n$ vertices represented by the $q^n$ $n$-tuples and edges by the $q^{n+1}$ words
of length $n+1$ over $\Sigma$. The edge $(x_0,x_1,\ldots,x_{n-1},x_n)$ is from the
vertex $(x_0,x_1,\ldots,x_{n-1})$ to the vertex $(x_1,x_2,\ldots,x_n)$.
Other sequences (paths and cycles) in the graph are also very important for theoretical
and practical points of view~\cite{golomb2017shift}. Any sequence $\cS = s_0 s_1 s_2 ~ \cdots$ over $\Sigma$
can be considered as a sequence generated by a path in $G_{q,n}$, where each window of length $n$
is associated with a vertex of $G_{q,n}$ and two consecutive windows of length $n$
are connected by a directed edge from the first vertex (window) to the second vertex.

A {\bf \emph{covering code of length $n$ and radius $R$}} over $\Sigma$,
consists of codewords of length $n$ over $\Sigma$, where the union of the balls
of radius $R$ around these codewords covers the entire space~$\Sigma^{n}$,
where the ball or radius $R$ around a word $(x_1,x_2,\ldots,x_n)$ consists of all the words
of length $n$ over $\Sigma$ whose Hamming distance from $(x_1,x_2,\ldots,x_n)$ is at most $R$.
Such a ball has size $V_q(n,R)$, where
$$
V_q(n,R) = \sum_{i=0}^{R} \binom{n}{i}(q-1)^{i} ~.
$$
The balls of radius $R$ around the codewords of such a covering code $\cC$ must contain the
whole space $\Sigma^n$. Hence, a lower bound on the size of such covering code $\cC$ is
$$
\abs{\cC} \geq \frac{q^n}{V_q (n,R)}.
$$
This bound is known as the sphere covering bound. On the other hand, there is a covering code of radius $R$ achieving this bound up to a multiplicative constant that depends on $R$ and~$q$~\cite{krivelevich2003covering}.
Covering codes have numerous applications in various domains,
e.g.,~\cite{cohen1997covering,smolensky1993representations,pagh2016locality,afrati2014anchor,lenz2020covering}.
An excellent book on covering codes is the manuscript~\cite{cohen1997covering}.

It is quite natural to combine these two concepts of sequences in the de Bruijn graph and covering codes
as was done in~\cite{chung2004bruijn}. An {\bf \emph{$(n,R)$-de Bruijn covering sequence}} (dBCS in short)
is a cyclic sequence over $\Sigma$, in which the union of all substrings of length $n$ forms
a code of length $n$ and covering radius $R$ over $\Sigma$. We prefer to call
this concept a sequence rather than a code (as was done in~\cite{chung2004bruijn}) as there is a unique
sequence in this code and the consecutive $n$-tuples form the code.
Following the work of~\cite{chung2004bruijn} which considered an alphabet whose
size is a power of a prime, a probabilistic upper bound on the length of a covering sequence for any
alphabet was presented in~\cite{vu2005bruijn}, where it was proved that
there exists a de Bruijn covering sequence whose length is at most
$\bigO\left(\frac{q^n}{V_q(n,R)}\log n \right)$ for fixed $q$ and $R$.

In recent years many one-dimensional coding problems have been considered in the two-dimensional framework due to
modern applications, e.g.,~\cite{Rot91,TER09,etzion1988constructions}.
This includes many structures associated with de Bruijn sequences and other
sequences in the de Bruijn graph. Such structures include de Bruijn arrays known also as perfect maps~\cite{Pat94},
pseudo-random arrays~\cite{McSl76},
robust self-location arrays with window property~\cite{bruckstein2012simple},
structured-light patterns~\cite{morano1998structured}.
Therefore, it is very tempting to generalize the concept of de Bruijn covering sequences
and this is one of the targets of the current work.
An $(m,n,R)$-de Bruijn covering array (dBCA in short) is a doubly periodic $M \times N$ array
over an alphabet of size $q$ such that the set
of all its $m \times n$ windows form a covering code with radius~$R$. While in the one-dimensional
case we are interested in the $(n,R)$-dBCS of the shortest length, in the two-dimensional
case we are interested in in the $(m,n,R)$-dBCA with the smallest area. However, also the
actual dimensions of the array $M \times N$ will be of significant importance.

The rest of the paper is organized as follows. In Section~\ref{sec:prob_method} we follow
the method of Vu~\cite{vu2005bruijn} and obtain an upper bound on the possible areas of de Bruijn
covering arrays using probabilistic arguments. The adaptation of the method to the two-dimensional framework
requires some more delicate computations. Unfortunately, these probabilistic arguments imply
only existence results and not constructive ones.
In Section~\ref{sec:construction} we construct de Bruijn covering arrays based on dBCSs.
The technique is based on folding the one-dimensional sequences into two-dimensional arrays
and tiling these arrays. The area of the constructed arrays
is either about twice or at least twice the length of the associated sequences.
Section~\ref{sec:one-two-dimensional} presents several methods to construct dBCSs.
Many of the sequences constructed by these methods are much better (shorter in length) than the known sequences.
Moreover, some of the methods yield sequences that can be efficiently generated.
The methods also provide upper bounds on the length which do not involve probabilistic methods and arguments.
A few interesting suggested methods are based on sequences from cyclic codes,
primitive polynomials, an interleaving technique, and self-dual sequences.
We also examine the construction of several sequences whose\break $n$-tuples form
a covering code with radius $R$. These sequences can be used in the construction of two-dimensional arrays.
Section~\ref{sec:two-dimensional} will devoted to presenting several constructions for dBCAs with small
parameters. The constructions will be mainly based on folding one-dimensional sequences or
using appropriate shifts of one-dimensional sequences.

 \section{A probabilistic construction}
\label{sec:prob_method}

Does a dBCA whose area is close to the sphere covering bound exist?
In this section, using a similar approach as in~\cite{vu2005bruijn}, but with a slightly more delicate analysis,
we show such an upper bound for a small area of a dBCA.
\begin{theorem}
\label{thm:prob_bound}
Let $m,n$ be nonnegative integer numbers. For any $M\geq m$, there exists an $M \times N$ $(m,n,R)$-dBCA
such that  $MN = \bigO \left( \frac{q^{mn}}{V_q(mn,R)} \cdot (\log m + \log n)  \right)$,
for fixed $q$ and $R$. 
\end{theorem}
Before proving Theorem~\ref{thm:prob_bound}, we first present the following primary result. 
For an array $A = (a_{i,j}) \in \Sigma^{M \times N}$, let $C_A$ be the set of elements that are at a distance
of at most $R$ from~$S_{A}$, where $S_A$ is the set of all windows of size $m \times n$ of $A$.
The following lemma can be easily verified.
\begin{lemma}
\label{claim:covering}
If for an array $A \in \Sigma^{M \times N}$, $C_A$ does not contain exactly $L$ elements
of $\Sigma^{m \times n}$, then there exists an $(m,n,R)$-dBCA $A'$ whose
area is $\bigO \left(M \cdot (N+\frac{mnL}{M}) \right)$.  
\end{lemma}
\begin{proof}[Proof of Theorem~\ref{thm:prob_bound}]
We prove the theorem using the probabilistic method.
A random array $A = (a_{i,j}) \in \Sigma^{M\times N}$ (where $N$ will be chosen later)
is constructed as follows. For each $(i,j) \in [M] \times [N]$, $a_{i,j}$ takes a uniformly
random value in $\Sigma$. For each $v \in \Sigma^{m \times n}$, let $E_{[v]}$ be the event that
the element $v$ is not contained in $C_A$ and denote $L = \sum_{v \in \Sigma^{m \times n}}\Pr(E_{[v]})$.
Then $L$ is the expected number of elements of $\Sigma^{m \times n}$ not contained
in $C_A$. By the linearity of expectation and symmetry, we have that
\begin{align*}
L = \sum_{v \in \Sigma^{m \times n}}\Pr(E_{[v]}) = q^{mn} \Pr(E_{[0]}) \, ,
\end{align*}
where $[0] \in \Sigma^{m \times n}$ is the all-zeros matrix.
\par 
Let $V$ be the set of elements at a distance of at most $R$ from $[0] \in \Sigma^{m \times n}$.
The event $E_{[0]}$ occurs if and only if $S_A$ does not intersect $V$.
For each $(i,j) \in [M] \times [N]$, denote by $B_{i,j}$ the event that $s_{i,j} \in S_{A}$ belongs to $V$,
where $s_{ij}$ represents the window of $A$ with its upper-left corner at position $(i,j)$. For these definitions we have that
\begin{align*}
\Pr (E_{[0]}) = \Pr \left( \bigwedge_{i,j =1}^{M,N} \overline{B}_{i,j} \right) \, .
\end{align*}
For $(i, j), (\ell, k) \in [M] \times [N]$, $(i, j) \neq (\ell, k)$ we denote ${(i, j) \sim (\ell, k)}$
if $s_{i,j}$ and $s_{\ell,k}$ intersect. We define
$$
\mu = \sum_{i=1}^{M}\sum_{j=1}^{N} \Pr(B_{i,j}) \, ,
$$
$$\Delta = \sum_{(i,j) \sim (\ell,k)} \Pr(B_{i,j} \wedge B_{\ell,k}),
$$
and
$$
\delta = \sum_{(1,1) \sim (i,j)} \Pr(B_{i,j}).
$$
Using Suen's inequality~\cite[Theorem 3]{janson1998new}, we have that
\begin{align*}
\Pr \left( \bigwedge_{i,j =1}^{M,N} \overline{B}_{i,j} \right) \leq \exp \left( -\min \left(\frac{\mu^2}{8\Delta}, \frac{\mu}{2}, \frac{\mu}{6\delta} \right)  \right) \, .
\end{align*}
    
We will estimate $\mu,\Delta$ and $\delta$. First, for $\mu$, by linearity of expectation and symmetry, one has
\begin{align*}
    \mu = MN\cdot \Pr(B_{1,1}) = MN \cdot \frac{V_q(mn,R)}{q^{mn}}  \, .
\end{align*}
Next, for any $(i,j) \in [M] \times [N]$ there are $(2m-1)(2n-1)-1$ pairs $(\ell,k)$ such
that $(i,j) \sim (\ell,k)$. Therefore, 
\begin{align*}
\delta = (4mn -2m-2n)\Pr(B_{1,1}) \leq 4mn \frac{V_q(mn,R)}{q^{mn}}  = o(1)\,.
\end{align*}
For estimating $\Delta$, we have the following lemma whose proof can be found in the Appendix.
\begin{lemma}
\label{lem:Delta_bounds}
We have $\Delta \leq (1 + o(1))\mu$. 
\end{lemma}

From the estimation of $\Delta,\mu$, and $\delta$, we have that
\begin{align*}
\min \left(\frac{\mu^2}{8\Delta}, \frac{\mu}{2}, \frac{\mu}{6\delta} \right) \geq c\mu \, \text{ for some positive constant c}.
\end{align*}
Therefore, 
\begin{align*}
    \Pr\left( \bigwedge_{i,j =1}^{M,N} \overline{B}_{i,j} \right) \leq \exp \left( -c\mu \right) \, .
\end{align*}
By choosing $N$ such that
\begin{align*}
    MN &= \frac{1}{c} \frac{q^{mn}}{V_q(mn,R)}\log(mn \cdot V_q(mn,R)) \,.
\end{align*}
Now, we have that 
\begin{align*}
    \Pr(E_{[0]}) = \Pr\left( \bigwedge_{i,j =1}^{M,N} \overline{B}_{i,j} \right) \leq \frac{1}{mn \cdot V_q(mn,R)} \, .
\end{align*}
Therefore, 
\begin{align*}
    L = q^{mn} \Pr(E_{[0]}) \leq \frac{q^{mn}}{mn \cdot V_q(mn,R)} \, .
\end{align*}
On the other hand, since $q$ and $R$ are fixed, we have that $\log(mn \cdot V_q(mn,R)) = \bigO \left( \log m + \log n  \right)$. Using Lemma~\ref{claim:covering}, the proof of the theorem is completed. 
\end{proof}


\section{Folding of a dBCS into a dBCA}
\label{sec:construction}

In this section, we present constructions for dBCAs by applying folding on dBCSs.
We start by folding an $(mn,R)$-dBCS into an $(m,n,R)$-dBCA.
The given $(mn,R)$-dBCS $\cS=s_0,s_1,\ldots,s_{k-1}$ of length $k$, where $k$ is divisible by~$n$
(if it is not divisible by~$n$ we can extend it by at most $mn$ symbols to be divisible by $n$ and
remain a dBCS.).
Each substring of~$\cS$ whose length is $mn$ will be folded row by row into an associated $m \times n$
window of an $M \times N$ array. This will imply immediately that the constructed $M \times N$ array
is an $(m,n,R)$-dBCA. Each entry $s_\ell$ of the $(mn,R)$-dBCS will be in the
upper-left corner of exactly one $m \times n$
array which contains the $mn$ consecutive entries of~$\cS$ that start in~$s_\ell$ folded row by row.

To achieve these goals the $j$th row of the array, where\break
$0 \leq j \leq \frac{k}{n} -1$, is defined by
\begin{equation}
\label{eq:row_in_array}
s_{jn+1},s_{jn+2},\ldots,s_{jn+n},s_{jn+n+1},\ldots,s_{jn+2n-1},
\end{equation}
where the indices are taken modulo $k$.

Consider now $m$ consecutive rows of the defined array starting at row $j$, where
$0 \leq j \leq \frac{k}{n} -1$ and $n$ consecutive columns starting at column~$i$, where $1 \leq i \leq n$.
These rows and columns define the following $m \times n$ window of the array.

$$
\begin{array}{cccc}
s_{jn+i} & s_{jn+i+1} &  \cdots  & s_{jn+i+n-1} \\
s_{(j+1)n+i} & s_{(j+1)n+i+1} &  \cdots & s_{(j+1)n+i+n-1} \\
\vdots & \vdots & \ddots & \vdots \\
s_{(j+m-1)n+i} & s_{(j+m-1)n+i+1} &  \cdots &  s_{(j+m-1)n+i+n-1},
\end{array}
$$
where indices are taken modulo $k$.

This construction implies the following consequence
\begin{theorem}
If $\cS$ is an $(mn,R)$-dBCS of length $k$, where $n$~divides~$k$, then there exists a doubly periodic
$(m,n,R)$-dBCA of size $M \times N$, where $M= \frac{k}{n}$ and $N=2n-1$.

If $\cS$ is an $(mn,R)$-dBCS of length $k$, where $n$ does not divide $k$, then there exists a doubly periodic
$(m,n,R)$-dBCA of size $M \times N$, where $N=2n-1$ and $M= \left\lceil \frac{k}{n} \right\rceil +m-1$.
\end{theorem}

Next, we will be interested in $M \times N$ $(m,n,R)$-dBCAs for which $M$ will be considerably larger than $m$
and $N$ considerably larger than $n$ and the ratio $\frac{MN}{k}$ is small as possible, where there exists
an $(nm,R)$-dBCS sequence\break $\cS=\{s_0,s_1,s_2,\ldots,s_{k-1} \}$ whose length is $k$. We partition the
sequence $\cS$ into $t \cdot r$ sequences of length $\kappa \geq \left\lceil \frac{k}{t \cdot r} \right\rceil$.
Let $x_0,x_1,x_2,\ldots,x_{\kappa-1}$ be one such a sequence. From such a sequence we form
an $M' \times N'$ array, where $N'=2n-1$ and $M' \leq \left\lceil \frac{\kappa}{n} \right\rceil +m-1$.

The $j$th row of the array is defined by
$$
x_{jn+1},x_{jn+2},\ldots,x_{jn+2n-1},
$$
where indices are taken modulo $\kappa$ and each such array is constructed in the same way in which
one array was constructed before in~(\ref{eq:row_in_array}).

Once all the $M' \times N'$ arrays are constructed in this way, we form the $M \times N$ array,
where $M=rM'$ and $N=tN'$ by tiling this array with the $rt$ arrays of size $M' \times N'$ in any
possible way. In each such array $\kappa$ $m \times n$ windows are covered by folding an associated
part of the sequence row by row. The rightmost $n-1$ columns are required since the constructed
array from each section of $\cS$ is not periodic. For the same reason, the last $m-1$ rows
of each such array are required. From each $M' \times N'$ such sub-array whose area is $M' (2n-1)$,
at most $M' (n-1) + (M'- \frac{\kappa}{n})n$ bits are redundant. As $\kappa$ gets larger, the redundancy
bits are about half of the bits in the $M' \times N'$ sub-array and as a consequence
also about half of the bits in the $N \times N$ array.

This method can be applied to obtain $M \times N$ $(m,n,R)$-dBCAs for various values of $M$ and $N$.
The tiling is not necessarily done with sub-arrays of the same size. Asymptotically, this method
might be the best one to obtain such arrays with various parameters. Another way to apply the same
idea is to use several sequences that together contain $n$-tuples that form a covering
code with radius $R$. A set of such sequences will be called an $(n,R)$-de Bruijn covering sequences code
(or dBCSC in short). Such sequences are also discussed in the next section.
  
\section{Constructions for dBCSs and dBCSCs}
\label{sec:one-two-dimensional}

In this section, we will concentrate on constructing dBCSs
and dBCSCs. The section will consider only binary sequences,
but the methods work also on any non-binary alphabet.
One of the most attractive parameters that is usually the
first to consider is radius 1. Chung and Cooper~\cite{chung2004bruijn} used computer search to find
a short sequence as possible. Up to length ${n=8}$, the shortest dBCS was
found. From length $n=9$ the unrealistic smallest size of covering code was used as a lower bound
in~\cite{chung2004bruijn}. 
Improving the lower bound and having a general lower bound is an interesting question that will be
discussed in the full version of the paper.

\vspace{0.2cm}

\noindent
{\bf Construction from primitive polynomials:}

\vspace{0.1cm}

Let $c(x) = \sum_{i=0}^n c_i x^i$, where $c_n=c_0=1$ be an irreducible polynomial. Define the following sequence
\begin{equation}
\label{eq:rec_Mseq}
a_k = \sum_{i=1}^n c_i a_{k-i}
\end{equation}
with the initial nonzero $n$-tuple $( a_{-n},a_{-n+1},\ldots,a_{-1} )$.
If $c(x)$ is a primitive polynomial, then the sequence\break $\cA=( a_0,a_1,a_2 ,\ldots )$ is an M-sequence
of period $2^n-1$ in which each nonzero $n$-tuple appears exactly once as a window of length $n$.
Consider now another recursion for a sequence defined by
\begin{equation}
\label{eq:rec_comp_Mseq}
b_k = \sum_{i=1}^n c_i b_{k-i} +1
\end{equation}
with the initial nonzero $n$-tuple $( b_{-n},b_{-n+1},\ldots,b_{-1} )$.

\begin{lemma}
The sum $\sum_{i=1}^n c_i$ is an even integer.
\end{lemma}

\begin{lemma}
The sequence $\cB=(  b_0,b_1,b_2 ,\ldots )$ is the binary complement of the sequence $\cA$.
\end{lemma}
\begin{proof}
Let $(a_j,a_{j+1},\ldots,a_{j+n-1})$ be an $n$-tuple in the sequence $\cA$. By Eq.~(\ref{eq:rec_Mseq})
we have that
$$
a_{j+n}= \sum_{i=1}^n c_i a_{j+n-i}.
$$
Let $(b_j,b_{j+1},\ldots,b_{j+n-1})$ be an $n$-tuple in the sequence $\cB$, where $b_i=a_{i}+1$,
for $j \leq i \leq j+n-1$. Therefore, we have that
\begin{align*}
b_{j+n} & = \sum_{i=1}^n c_i b_{j+n-i}+1 &= \sum_{i=1}^n c_i (a_{j+n-i}+1) +1 \\
&= \sum_{i=1}^n c_i a_{j+n-i} + \sum_{i=1}^n c_i +1 &= \sum_{i=1}^n c_i a_{j+n-i}+1 ~~~~~~~ \\
& =a_{j+n} +1 . & 
\end{align*}
Therefore, the sequence $\cB$ is the binary complement of the sequence~$\cA$.
\end{proof}

\begin{corollary}
The recursion of Eq.~(\ref{eq:rec_Mseq}) generates the sequence~$\cA$ and the all-zeros sequence.
The recursion of Eq.~(\ref{eq:rec_comp_Mseq}) generates the sequence $\cB=\bar{\cA}$ and the all-ones sequence.
\end{corollary}

\begin{lemma}
\label{lem:prim_cov}
Let $c(x) = \sum_{i=0}^n c_i x^i$ be a primitive polynomial for which $c_i=0$ for
$1 \leq i \leq 2R+1$ and consider the sequences $\cA$, $\cB$, the all-zeros sequence
and the all-ones sequence. The code that contains these four sequences is an $(n+2R+1,R)$-dBCSC.
\end{lemma}
\begin{proof}
Let $X$ be any given $n$-tuple.
Since the first $2R+1$ $c_i$s (except for $c_0$) are \emph{zeros}, it follows that the last $2R+1$
elements of $X$ are not influencing the result of the next bit for both recursions and hence
the addition of the 1 in the sequence~$\cB$ implies that the next $2R+1$ bits after $X$ in $\cA$
and~$\cB$ will be complementing. This implies that $X z_1 z_2 ~ \cdots ~ z_{2R+1}$ is in the ball of radius $R$ either
by an $(n+2R+1)$-tuple of~$\cA$ that starts with $X$ or  by an $(n+2R+1)$-tuple of $\cB$ that starts with~$X$.
\end{proof}

The number of sequences of the code defined in Lemma~\ref{lem:prim_cov} is four and they can be efficiently
concatenated to one ${(n+2R+1,R)}$-dBCS.
One advantage that the construction based on a primitive polynomial has on some other constructions
is that the generated sequence can be obtained with the recursions of the sequences $\cA$ and $\cB$
in an efficient way.

The analysis of the construction of sequences that are generated by a primitive polynomial can
be done for any associated irreducible polynomial to generate an $(n+2R+1,R)$-dBCSC.
This will enable to obtain $(n+2R+1,R)$-dBCSCs with various parameters and
they can be used to construct dBCAs by using the folding techniques of Section~\ref{sec:construction}.


\noindent
{\bf Construction from cyclic covering codes:}


This is one of the most effective constructions. Let $\cC$ be a cyclic code of length $n$ and covering radius $R$.
Since $\cC$ is cyclic, it follows that we can partition the codewords into equivalence classes, where
two codewords are in the same equivalence class if one is a cyclic shift of the other.
The size of an equivalence class is a divisor of $n$ and usually most classes are of size $n$.
For each equivalence class of size $d$ we construct a string of length $n+d-1$.
All these strings are now combined with possible overlap between them to reduce the length
of the dBCS. Cyclic covering codes were considered
for example by~\cite{dougherty1991covering,downie1985covering,janwa1989some,kavut2019covering}.

\begin{example}

For $n=15$, we consider the cyclic Hamming code of length 15.
It has 134 classes of size 15, 6 classes of size 5, 2 classes of size 3, and 2 classes of size 1.
They are combined to form a $(15,1)$-dBCS of length 3600 compared to a lower bound 2048 obtained by the
sphere covering bound.

\hfill\quad $\blacksquare $
\end{example}


\noindent
{\bf Construction from self-dual sequences:}

\vspace{0.1cm}

This construction is demonstrated for values of $n$ which are powers of 2.
We start with the optimal $(8,1)$-dBCS of length~32,
$$
\cS_1 , \cS_2 = 0001101111100100 , 0001101011100101
$$
$\cS_1$ and $\cS_2$ are two self-dual sequences. $\cS_1$ has the form $[X ~ \bar{X}]$ and
$\cS_2$ has the form $[Y ~ \bar{Y}]$, where $Y$ differs from $X$
only in the last bit. For each string $Z$ of even weight and length~8, which starts
with a \emph{zero} the following two sequences are constructed.
$$
[Z , Z+X , \bar{Z} , \bar{Z} + X ] ~~~~~ [Z , Z+Y , \bar{Z} , \bar{Z} + Y ]~.
$$
In this way, we form 64 pairs of cyclic sequences of length~32. Each pair can be merged into another sequence
$$
[Z , Z+X , \bar{Z} , \bar{Z} + X , Z , Z+Y , \bar{Z} , \bar{Z} + Y ]~.
$$
The idea of constructing such self-dual sequences appeared already in~\cite{etzion1984construction,etzion1996near}
for constructions of other types of sequences.
These 64 sequences form a $(16,1)$-dBCSC. Trivially, each such sequence can be extended
with its first 15 bits to an acyclic sequence. These 64 sequences can be concatenated to form
a $(16,1)$-dBCS of length ${64 \cdot (64+15)=5056}$. We can also merge the sequences with overlaps
of suffixes and prefixes of sequences and obtain a $(16,1)$-dBCS of length~4476 (the known lower bound
is 4096). The same idea can be applied recursively now on the 64 pairs of sequences of length~32
to obtain an upper bound on the shortest length of a $(2^k,1)$-dBCS. For $n=2^k$, $k>2$, the upper bound is considerably
less than $2^{2^k -k} + 2^{2^k -k-2}$ compared to the lower bound~$2^{2^k -k}$.


\noindent
{\bf Interleaving of two dBCSs:}


Interleaving of two dBCSs sequences, one with radius $R_1$ and a second one with radius $R_2$
can lead to a short sequence with radius $R_1 + R_2$. Let $\cS = s_0,s_1,s_2,\ldots,s_{k_1-1}$ be an $(n_1,R_1)$-dBCS
of length $k_1$ and $\cT = t_0,t_1,t_2,\ldots,t_{k_2-1}$ be an $(n_2,R_2)$-dBCS of length $k_2$,
where $n_1 -1 \leq n_2 \leq n_1 +1$ and g.c.d.$(k_1,k_2)=1$. The interleaving of $\cS$ and $\cT$
(writing an element from one sequence and one from the other cyclically until the beginning of both sequences
repeats) is an $(n,R)$-dBCS of length $k=k_1k_2$, where $n=n_1+n_2$ and $R_1 +R_2$.

\begin{example}
\label{ex:inter_r=1}

Consider the $(5,1)$-dBCS of length 8, $10100011$, and the $(5,0)$-dBCS of length 33 (a de Bruijn sequence
of length 32 to which an extra \emph{zero} is added to the run of 5~\emph{zeros}). Their interleaving yields
a $(10,1)$-dBCS of length 264 which is shorter than the
$(10,1)$-dBCS of length~322 found in~\cite{chung2004bruijn} using computer search.

Consider the $(6,1)$-dBCS of length 12, $000100111011$, and the $(5,0)$-dBCS of length 35 (a de Bruijn sequence
of length~32 to which extra three \emph{zeros} are added to the run of 5~\emph{zeros}). Their interleaving yields
a $(11,1)$-dBCS of length~420 which is shorter than the
$(11,1)$-dBCS of length~694 found in~\cite{chung2004bruijn} using computer search.

Consider the $(6,1)$-dBCS of length 12, $000100111011$, and the $(6,0)$-dBCS of length 65 (a de Bruijn sequence
of length~64 to which an extra \emph{zero} is added to the run of 6~\emph{zeros}). Their interleaving yields
a $(12,1)$-dBCS of length 780 which is shorter than the
$(12,1)$-dBCS of length 1454 found in~\cite{chung2004bruijn} using computer search.

Consider the $(6,0)$-dBCS of length 65 (a de Bruijn sequence
of length 64 to which an extra \emph{zero} is added to the run of 6~\emph{zeros})
and the $(7,1)$-dBCS of length~22, $1111001010110010000110$. Their interleaving yields
a $(13,1)$-dBCS of length 1430 which is shorter than the
$(13,1)$-dBCS of length 2937 found in~\cite{chung2004bruijn} using computer search.

\hfill\quad $\blacksquare $
\end{example}

\begin{example}
\label{ex:interleave}

Consider the $(6,1)$-dBCS of length 12, $000100111011$, and the $(6,1)$-dBCS of length 17,
$00000010101111011$. Their interleaving yields a $(12,2)$-dBCS of length 204 which is shorter than the
$(12,2)$-dBCS of length 244 found in~\cite{chung2004bruijn} using computer search.

Consider the $(6,1)$-dBCS of length 12, $000100111011$, and the $(7,1)$-dBCS of length 25,
$1111110101100000101001100$. Their interleaving yields a $(13,2)$-dBCS of length 300 which is shorter than the
$(13,2)$-dBCS of length 529 found in~\cite{chung2004bruijn} using computer search.

Consider the $(7,1)$-dBCS of length 22, $1111001010110010000110$, and the $(7,1)$-dBCS of length~25,
$1111110101100000101001100$. Their interleaving yields a $(14,2)$-dBCS of length 550.

\hfill\quad $\blacksquare $
\end{example}
Finally, we should state that if the two interleaved sequences can be generated efficiently, then
also, the constructed sequence can be generated efficiently. In addition, 
using other methods, which will be described in the full version of this paper more bounds were improved. 
A table with the bounds on the length
of the shortest $(n,R)$-dBCS for $n$ up to~20 and $R=1,2$ is given in Table~\ref{table:one_dim}.

\begin{table}[h!]
	\centering
	\small
	\begin{tabular}{|c|c|c|c|c|c|c|c|}
		\hline
		\multirow{3}{*}{$n$} & \multicolumn{2}{c|}{Chung \& Cooper~\cite{chung2004bruijn}} & %
		\multicolumn{2}{c|}{Interleaving method} & \multicolumn{1}{c|}{other}\\
		
		\cline{2-6}
		& $R=1$ & $R=2$ & $R=1$ & $R=2$&$R$=1   \\ 
            \hline
		$9$&130 &20 & 136& 34 & 101\\
		\hline
		$10$&322 &38 & 264&66 & 180\\
		\hline
		$11$& 694& 117& 420&120 & 288\\
		\hline
		$12$& 1454& 244& 780&204 & 632 \\
		\hline
		$13$& 2937& 529& 1430&300 & 1205\\
		\hline
		$14$& -& -&2838 &550 &2292 \\
		\hline
		$15$& -&- &4128 &770 & 3600\\
		\hline
		$16$& -& -& 8224&1120 & 4476\\
		\hline
		$17$& -&- & 25856& 3535 & \\
		\hline
		$18$& -&- &51712 & 10260 & \\
		\hline
            $19$& -&- &103424 & 19494 & \\
		\hline
            $20$& - &-  &184500 & 32580 &  \\
		\hline
	\end{tabular}
\caption{Upper bounds for the length of the shortest $(n,R)$-dBCS for $9 \leq n \leq 20$ and $R=1,2$.}
\label{table:one_dim}
\end{table}

\section{Small Area dBCAs with small parameters}
\label{sec:two-dimensional}

In this section, we will describe a few constructions for $M \times N$ $(m,n,R)$-dBCAs.
The target will be to have an area as small as possible, but we want various selections of parameters $M$ and $N$ and
especially we do not want $M$ to be too close to $m$ or $N$ too close to $n$, although for completeness also
such arrays will be mentioned. Most of the constructions will be general, but their effectiveness asymptotically
will be doubtful. Nevertheless, we should mention that currently, our asymptotic results are probabilistic, i.e., the
existence of such an array is proved, but we have no idea how to construct these arrays.

Using the results from Section~\ref{sec:one-two-dimensional} for generating dBCSs and employing the
construction in Section~\ref{sec:construction} to generate dBCAs from given dBCSs,
we provide several $(m,n,R)$-dBCAs with small area sizes.

\vspace{0.2cm}

\noindent
{\bf A $(2,n,2R)$-dBCA from an $(n,R)$-dBCS:}

\vspace{0.1cm}

Let $\cS$ be an $(n,R)$-dBCS of length $k$.

If $k$ is even, then form an $(k+1) \times k$ array
whose $i$th row, $0 \leq i \leq k-1$ is ${\bf E}^j S$, $j = \sum_{\ell=0}^i \ell$,
where ${\bf E}^jS$ is a cyclic shift of $\cS$
by $j$ positions, i.e.,
$$
{\bf E}^j (s_0,s_1,\ldots, s_{k-1}) = (s_j,s_{j+1},\ldots,s_{k-1},s_0,\ldots,s_{j-1})
$$
and the $k$th row is the same as the $(k-1)$th row.

If $k$ is odd, then form an $k \times k$ array whose $i$th row,\break
$0 \leq i \leq k-1$ is ${\bf E}^j S$, $j = \sum_{\ell=0}^i \ell$.

\begin{example}
\label{ex:multi_shift}

Consider the $(6,1)$-dBCS of length 12 in Example~\ref{ex:interleave}. By applying the
defined shifts to this sequence which is the first row in the array we obtain a $13 \times 12$ $(2,6,2)$-dBCA
whose area 156 is less than the related length of the $(12,2)$-dBCS of length 204 presented
in Example~\ref{ex:interleave}.

Consider the $(7,1)$-dBCS of length 22 in Example~\ref{ex:interleave}. By applying the
defined shifts to this sequence which is the first row in the array we obtain a $23 \times 22$ $(2,7,2)$-dBCA
whose area 506 is less than the related length of the $(14,2)$-dBCS of length 550 presented
in Example~\ref{ex:interleave}.

\hfill\quad $\blacksquare $
\end{example}

It is trivially observed that applying the folding on an $(mn,R)$-dBCS of length $k$, yields an $M \times N$ $(m,n,R)$-dBCA whose
area $M \cdot N$ is at least $m(2n-1)$, but usually $M \cdot N > 2k$.
On the other hand, applying the shifts on an $(n,R)$-dBCS of length $N$ as presented in the construction and demonstrated
in Example~\ref{ex:multi_shift} yields an $M \times N$ $(2,n,2R)$-dBCA whose area is at most the length
of the $(2n,2R)$-dBCS obtained by applying interleaving on $(n,R)$-dBCSs. This immediately raises an interesting
question on the ratio between the length of an $(mn,2R)$-dBCS and the area of an $(m,n,2R)$-dBCA. 

A table with the parameters of some dBCAs is given in Table~\ref{table:two_dim} (see the Appendix).


\newpage

\begin{thebibliography}{99}
\bibitem{de1946combinatorial}
N.~G.~De Bruijn, ``A combinatorial problem," \emph{Proceedings of the Section
of Sciences of the Koninklijke Nederlandse Akademie van Wetenschappen te Amsterdam,} vol.~49, no.~7, 
pp.~758--764, 1946.
\bibitem{IJgood46}
I.~G.~Good, ``Normally recurring decimals," \emph{J. London Math. Soc.,}
vol.~21, pp.~167--169, 1946.
\bibitem{dBEh51a}
T.~van Aardenne-Ehrenfest and N.~G.~de Bruijn, ``Circuits and trees in ordered linear graphs," \emph{Simon Steven,} vol.~28, pp.~203--217, 1951.
\bibitem{lempel1985design}
A. Lempel and M. Cohn, “Design of universal test sequences for vlsi,”
IEEE Transactions on Information Theory, vol.~31, no.~1, pp.~10--17,
1985.
\bibitem{golomb2017shift}
S.~W.~Golomb, ``Shift register sequences: secure and limited-access code
generators, efficiency code generators, prescribed property generators,
mathematical models" \emph{World Scientific,} 2017.
\bibitem{bruckstein2012simple}
A.~M.~Bruckstein, T.~Etzion, R.~Giryes, N.~Gordon, R.~J.~Holt, and D.~Shuldiner, ``Simple and robust binary self-location patterns," \emph{IEEE
Transactions on Information Theory,} vol.~58, no.~7, pp.~4884--4889, 2012.
\bibitem{samatham1989bruijn}
M.~R.~Samatham and D.~K.~Pradhan, ``The de bruijn multiprocessor
network: a versatile parallel processing and sorting network for vlsi," \emph{IEEE Transactions on Computers,} vol.~38, no.~4, pp.~567--581, 1989.
\bibitem{chee2022run}
Y.~M.~Chee, D.~T. Dao, T.~L.~Nguyen, D.~H.~Ta, and V.~K.~Vu, ``Run length
limited de bruijn sequences for quantum communications," in \emph{2022
IEEE International Symposium on Information Theory (ISIT),} pp.~264--269, 2022.
\bibitem{pages2005optimised}
J.~Pages, J.~Salvi, C.~Collewet, and J.~Forest, ``Optimised de bruijn
patterns for one-shot shape acquisition," \emph{Image and Vision Computing,}
vol.~23, no.~8, pp.~707--720, 2005.
\bibitem{krivelevich2003covering}
M.~Krivelevich, B.~Sudakov, and V.~H.~Vu, ``Covering codes with
improved density," \emph{IEEE Transactions on Information Theory,} vol.~49,
no.~7, pp.~1812--1815, 2003.
\bibitem{cohen1997covering}
G.~Cohen, I.~Honkala, S.~Litsyn, and A.~Lobstein, \emph{Covering codes},
Elsevier, 1997.
\bibitem{smolensky1993representations}
R.~Smolensky, “On representations by low-degree polynomials,” in
\emph{Proceedings of 1993 IEEE 34th Annual Foundations of Computer
Science,} pp.~130--138, 1993.
\bibitem{pagh2016locality}
R. Pagh, ``Locality-sensitive hashing without false negatives," in \emph{Proceedings of the twenty-seventh annual ACM-SIAM symposium on Discrete algorithms,} pp.~1--9, 2016.
\bibitem{afrati2014anchor}
F.~Afrati, A.~D.~Sarma, A.~Rajaraman, P.~Rule, 
S.~Salihoglu, and J.~Ullman, ``Anchor points algorithms for hamming and edit distance,"
in \emph{Proceedings of ICDT,} 2014.

\bibitem{lenz2020covering}
A. Lenz, C. Rashtchian, P. H. Siegel, and E. Yaakobi, ``Covering
codes using insertions or deletions," \emph{IEEE Transactions on Information
Theory,} vol. 67, no. 6, pp. 3376--3388, 2020.
\bibitem{chung2004bruijn}
F. Chung and J. N. Cooper, ``De bruijn cycles for covering codes,"
\emph{Random Structures \& Algorithms,} vol. 25, no. 4, pp. 421–431, 2004.
\bibitem{vu2005bruijn}
V. Vu, ``De bruijn covering codes with arbitrary alphabets," \emph{Advances
in Applied Mathematics,} vol. 34, no. 1, pp. 65--70, 2005.
\bibitem{Rot91}
R. M. Roth, ``Maximum-rank array codes and their application to
crisscross error correction," \emph{IEEE Transactions on Information Theory,}
vol. 37, no. 2, pp. 328--336, 1991.
\bibitem{TER09}
I. Tal, T. Etzion, and R. Roth, ``On row-by-row coding for 2d constraints," \emph{IEEE Transactions on Information Theory,} vol. 55, no. 8,
pp. 3565--3576, 2009.
\bibitem{etzion1988constructions}
T. Etzion, ``Constructions for perfect maps and pseudorandom arrays," \emph{IEEE Transactions on Information Theory,} vol. 34, no. 5, pp. 1308--1316, 1988.
\bibitem{Pat94}
K. G. Paterson, ``Perfect maps," \emph{IEEE Transactions on Information
Theory,} vol. 40, no. 3, pp. 743--753, 1994.
\bibitem{McSl76}
F. J. MacWilliams and N. J. A. Sloane, ``Pseudo-random sequences and
arrays," \emph{Proceedings of the IEEE,} vol. 64, 
pp. 1715--1729, 1976.
\bibitem{morano1998structured}
R. A. Morano, C. Ozturk, R. Conn, S. Dubin, S. Zietz, and J. Nissano,
``Structured light using pseudorandom codes," \emph{IEEE Transactions on
Pattern Analysis and Machine Intelligence,} vol. 20, no. 3, pp. 322--327,
1998.

\bibitem{janson1998new}
S. Janson, ``New versions of suen’s correlation inequality," \emph{Random
Structures and Algorithms,} vol. 13, no. 3--4, pp. 467--483, 1998.
\bibitem{dougherty1991covering}
 R. Dougherty and H. Janwa, ``Covering radius computations for binary
cyclic codes," \emph{Mathematics of computation,} vol. 57, no. 195, pp. 415--434, 1991.
\bibitem{downie1985covering}
D. Downie and N. Sloane, ``The covering radius of cyclic codes of
length up to 31 (corresp.)," \emph{IEEE Transactions on Information Theory,}
vol. 31, no. 3, pp. 446--447, 1985.

\bibitem{janwa1989some}
H. Janwa, ``Some new upper bounds on the covering radius of binary
linear codes," \emph{IEEE Transactions on Information Theory,} vol. 35, no. 1,
pp. 110--122, 1989.


\bibitem{kavut2019covering}
S. Kavut and S. Tutdere, ``The covering radii of a class of binary cyclic
codes and some bch codes," \emph{Designs, Codes and Cryptography,} vol. 87,
no. 2-3, pp. 317--325, 2019.
\bibitem{etzion1984construction}
T. Etzion and A. Lempel, ``Construction of de bruijn sequences of minimal complexity," \emph{IEEE Transactions on Information Theory,} vol. 30,
no. 5, pp. 705--709, 1984.

\bibitem{etzion1996near}
T. Etzion and K. G. Paterson, ``Near optimal single-track gray codes,"
\emph{IEEE Transactions on Information Theory,} vol. 42, no. 3, pp. 779--789,
1996.


\end{thebibliography}
\newpage
\appendix
\section*{Proof of Lemma~\ref{lem:Delta_bounds}}
\label{Apen:Lemma_delta}
Recall that, in this work, we assume $q$, $R$ are fixed and $m$, $n$ are sufficiently large. 
Let $I = \{(i,j): i, j \in [m] \times [n]\} \setminus \{(1,1)\}$. We have 
\begin{align*}
    \Delta  = 2MN \cdot \sum_{(i,j) \in I} \Pr (B_{1,1} \wedge B_{i,j}) \, .
\end{align*}
For any $(i,j) \in I$, then $s_{i,j}$ and $s_{1,1}$ share a common array $s'$ of size $(m-i+1) \times (n-j+1)$. Let $s'_{1} = s_1 \setminus s'$ and $s'_{i,j} = s_{i,j} \setminus s'$.  Let $k_1,k_2,k_3$ be the integer numbers of non-zero elements in $s',s'_1,s'_{i,j}$, respectively. The event $B_{1,1} \wedge B_{i,j}$ holds if and only if $k_1 +k_2 \leq R$ and $k_1 +k_3 \leq R$.   
\par 
Let $K$ be the set of all triples of nonnegative integers $(k_1,k_2,k_3)$ such that $k_1 + k_2 \leq R$ and $k_1 +k_3 \leq R$. We have $\Pr(B_{1,1} \wedge B_{i,j})$ is equal to the following quantity. 
\begin{align*}
  \sum_{(k_{1},k_{2},k_{3})\in K}\frac{A \cdot B \cdot C \cdot (q-1)^{k_{1}+k_{2}+k_{3}}}{q^{2mn-(m-i+1)(n-j+1)}}\, , 
\end{align*}
where, 
\begin{align*}
    A &= \binom{(m-i+1)(n-j+1)}{k_{1}}, \\
    B &= \binom{mn-(m-i+1)(n-j+1)}{k_{2}},\\
    C &= \binom{mn-(m-i+1)(n-j+1)}{k_{3}} \, .
\end{align*}
Therefore, $\Delta$ is equal to
\begin{align*}
    2MN \cdot \sum_{(k_{1},k_{2},k_{3})\in K}\sum_{(i,j)\in I}\frac{A \cdot B \cdot C \cdot (q-1)^{k_{1}+k_{2}+k_{3}}}{q^{2mn-(m-i+1)(n-j+1)}} \, .
\end{align*}
In the one-dimensional case, i.e., when $M=m=1$, we have that $B$ and $C$ do not depend on $n$ i.e., the cardinality of the non-overlapping part does not rely on the window size but only on $j$. However, in the two-dimensional case, both quantities $B$ and $C$ depend on $m$ and $n$, that is, the cardinality of the non-overlapping part depends on the window size. This property makes it difficult to directly apply the approach presented in~\cite{vu2005bruijn} for estimating $\Delta$ in the two-dimensional case. Here we present an approach for estimating $\Delta$.  

Let $u = m-i+1$ and $v = n-j+1$. Let $D = (q-1)^{k_1+k_2+k_3}$, then
\begin{align*}
     \Delta \leq 2MN \cdot \sum_{(k_1,k_2,k_3) \in K} \sum_{u=1}^{m}\sum_{v=1}^{n} \frac{\binom{uv}{k_1}\binom{mn-uv}{k_2}\binom{mn-uv}{k_3} \cdot D}{q^{2mn-uv}} \, .
\end{align*}

For each $(k_1,k_2,k_3) \in K$. Let 
\begin{align*}
    S(k_1,k_2,k_3) &= \sum_{u=1}^{m}\sum_{v=1}^{n} \frac{\binom{uv}{k_1}\binom{mn-uv}{k_2}\binom{mn-uv}{k_3} \cdot D}{q^{2mn-uv}}  \\ 
    &\leq \frac{\binom{mn}{k_1}}{q^{mn}} \sum_{u=1}^{m}\sum_{v=1}^{n} \frac{\binom{mn-uv}{k_2} \binom{mn-uv}{k_3}\cdot D}{q^{mn-uv}}\\
    & \leq \frac{\binom{mn}{k_1}}{q^{mn}} \sum_{u=1}^{m}\sum_{v=1}^{n} \frac{(mn-uv)^{k_2+k_3}\cdot D}{q^{mn-uv}} \\
    & \leq \frac{\binom{mn}{k_1}}{q^{mn}} \sum_{u=1}^{m}\sum_{v=1}^{n} \frac{[(m-u)(n-v)]^{k_2+k_3} \cdot D}{q^{(m-u)(n-v)}}\,.
\end{align*}
The last inequality follows from the fact that for fixed $q, t \in \NN$, the function $f(u)= u^{t}/q^{u}$ is a decreasing function, and $mn-uv \geq (m-u)(n-v)$ holds for all $u \leq m$ and $v \leq n$. Let $t = m-u, z = n-v$, we have
\begin{align*}
     S(k_1,k_2,k_3)& = \frac{\binom{mn}{k_1}}{q^{mn}} \sum_{u=1}^{m}\sum_{v=1}^{n} \frac{[(m-u)(n-v)]^{k_2+k_3}}{q^{(m-u)(n-v)}}\cdot D \\
    &\leq \frac{\binom{mn}{k_1}}{q^{mn}} \sum_{t = 0}^{m-1}\sum_{z=0}^{n-1} \frac{(tz)^{k_2+k_3}}{q^{tz}} \cdot D \\
    & \leq \frac{\binom{mn}{k_1}}{q^{mn}} \sum_{t = 1}^{m-1}\sum_{z=1}^{n-1} \frac{(tz)^{k_2+k_3}}{q^{t+z}} \cdot D \\
   &\leq \frac{\binom{mn}{k_1}}{q^{mn}} \left( \left(\sum_{t=1}^{m-1} \frac{t^{k_2+k_3}}{q^t} \right) \left(\sum_{z=1}^{n-1} \frac{z^{k_2+k_3}}{q^z} \right) \right)\cdot D \\ 
   &\leq \frac{\binom{mn}{k_1}}{q^{mn}} \left( \left(\sum_{t=1}^{\infty} \frac{t^{k_2+k_3}}{q^t} \right) \left(\sum_{z=1}^{\infty} \frac{z^{k_2+k_3}}{q^z} \right) \right) \cdot D\\
    & \leq  (1+o(1)) \frac{\binom{mn}{k_1}}{q^{mn}} \cdot D \\
    & = (1+o(1)) \frac{\binom{mn}{k_1}}{q^{mn}} \cdot (q-1)^{k_1+k_2+k_3} \, .
\end{align*}
The last inequality follows the fact that the series $\sum_{t=1}^{\infty} \frac{(t)^{k_2+k_3}}{q^{t}}$ converges. Therefore, 
if $k_1<R$, then $S(k_1,k_2,k_3) =  o \left( \frac{\binom{mn}{R}}{q^{mn}} \cdot (q-1)^R\right)$. Furthermore, since the cardinality of $K$ only depends on $R$. Thus, we have
\begin{align*}
    \sum_{(k_1,k_2,k_3) \in K} S(k_1,k_2,k_3) = (1+o(1))\frac{\binom{mn}{R}}{q^{mn}}(q-1)^{R} \, .
\end{align*}
\par 
We can now conclude that 
\begin{align*}
    \Delta &\leq (1+o(1))MN \cdot \frac{\binom{mn}{R}}{q^{mn}}(q-1)^{R} \\
    &= (1+o(1))MN \cdot \frac{V(mn,R)}{q^{mn}}\\
    &= (1+o(1))\mu \, .
\end{align*}

\begin{table}[ht!]
    \begin{center}
\begin{tabular}{ |c|c| } 
 \hline
 $(m,n,R)$& Existence $(m,n,R)$-dBCA of size $M \times N$\\
 \hline
   (2,2,1)& $2\times 3$, $3 \times 2$, $3 \times 3$  \\ 
 \hline
  (2,3,1) & $4 \times 3$, $5 \times 3$,$4 \times 4$, $4 \times 5$, $5 \times 4$ \\ 
 \hline
  (2,4,1) & $3 \times 16$, $8 \times 7$ \\ 
 \hline
  (2,5,1) & $3 \times 132$, $51 \times 9$ \\ 
 \hline
  (2,6,1) & $3 \times 390$, $130\times 11$ \\ 
 \hline
  (3,3,1) & $5\times 44$, $44 \times 5$, $12\times 15$ \\ 
 \hline
  (3,4,1) &  $5\times 260$, $195 \times 7$\\ 
 \hline
  (2,2,2) &  $2 \times 2$\\ 
 \hline
  (2,3,2) & $3\times 2$, $3 \times 3$\\ 
 \hline
  (2,4,2) & $4 \times 4$, $4 \times 5$, $5 \times 4$ \\ 
 \hline
 (2,5,2) & $3\times 19$, $8 \times 9$ ,$10\times 11$\\ 
 \hline
 (2,6,2) & $3\times 102$, $12\times 13$, $34\times 11$ \\ 
 \hline
 (3,3,2) &$5 \times 6$, $5\times 7$, $6 \times 5$, $7 \times 5$ \\ 
 \hline
(3,4,2) & $5\times 68$, $51\times 7$ \\ 
 \hline
\end{tabular}
\end{center}
\caption{The existence of a $(m, n, R)$-dBCA of size $M \times N$ for several small parameters}
\label{table:two_dim}
\end{table}

\end{document}